\documentclass[
    a4paper,
    12pt,
    DIV=11,
    headinclude=false,
    footinclude=false,
]{scrartcl}

\usepackage{libertinus}
\usepackage[T1]{fontenc}
\usepackage[utf8]{inputenc}
\usepackage[main=english]{babel}
\usepackage{csquotes}
\usepackage{microtype}

\usepackage{mathtools}
\mathtoolsset{centercolon,mathic}
\allowdisplaybreaks[3] 
\usepackage{amsthm}
\usepackage[libertinus,slantedGreek,varbb]{newtxmath}
\let\emptyset\varnothing
\usepackage[scr=boondoxo,cal=boondoxo]{mathalpha}
\usepackage{derivative}
\derivset{\odv}[sep-frac-fun=\,]

\let\Delta\varDelta

\let\Phi\varPhi

\let\Psi\varPsi

\let\Lambda\varLambda

\let\Gamma\varGamma

\let\Omega\varOmega

\newcommand{\mathup}[1]{\mathrm{#1}}
\newcommand{\quadtext}[1]{\quad\text{#1}\quad}

\newcommand{\Alignindent}{\hspace*{2em}&\hspace*{-2em}}

\newcommand{\sumstack}[2][]{\ifstrempty{#1}{\sum_{\substack{#2}}}{\smashoperator[#1]{\sum_{\substack{#2}}}}}

\DeclareMathOperator{\Dist}{dist}
\DeclareMathOperator{\Diam}{diam}

\DeclareMathOperator{\support}{supp}
\DeclareMathOperator{\tr}{tr}

\newcommand{\LdistSYM}{d}
\newcommand{\LdiamSYM}{\Diam}

\DeclarePairedDelimiter{\intervaloo}{\lparen}{\rparen}
\DeclarePairedDelimiter{\intervaloc}{\lparen}{\rbrack}
\DeclarePairedDelimiter{\intervalco}{\lbrack}{\rparen}
\DeclarePairedDelimiter{\intervalcc}{\lbrack}{\rbrack}

\DeclarePairedDelimiter{\abs}{\lvert}{\rvert}
\DeclarePairedDelimiter{\paren}{\lparen}{\rparen}
\DeclarePairedDelimiter{\norm}{\lVert}{\rVert}

\newcommand{\VERT}[1]{#1|\mkern-1.5mu#1|\mkern-1.5mu#1|}

\ExplSyntaxOn
\makeatletter

\MT_delim_default_inner_wrappers:n{\tnorm}

\NewDocumentCommand{\tnorm}{ s o m }{
    \IfBooleanTF{#1}{
    \MT_delim_tnorm_star_wrapper:nnn%
        {\VERT{\bgroup\left}}{#3}{\VERT{\aftergroup\egroup\right}}
    }{
        \IfValueTF{#2}{
            \@nameuse{MT_delim_tnorm_nostarscaled_wrapper:nnn}%
                {\VERT{\@nameuse {\MH_cs_to_str:N #2 l}}}
                {#3}
                {\VERT{\@nameuse {\MH_cs_to_str:N #2 r}}}
        }{
            \MT_delim_tnorm_nostarnonscaled_wrapper:nnn%
                {\VERT{}}
                {#3}
                {\VERT{}}
        }
    }
}

\makeatother
\ExplSyntaxOff

\DeclarePairedDelimiter{\commutator}{\lbrack}{\rbrack}

\DeclarePairedDelimiter{\List}{\{}{\}}
\DeclarePairedDelimiter{\ceil}{\lceil}{\rceil}

\DeclarePairedDelimiterXPP{\Distp}[1]{\Dist}{\lparen}{\rparen}{}{#1}
\DeclarePairedDelimiterXPP{\Diamp}[1]{\Diam}{\lparen}{\rparen}{}{#1}

\DeclarePairedDelimiterXPP{\dist}[1]{\LdistSYM}{\lparen}{\rparen}{}{#1}
\DeclarePairedDelimiterXPP{\diam}[1]{\LdiamSYM}{\lparen}{\rparen}{}{#1}

\DeclareDocumentCommand{\trace}{s o e{_} m}{
    \IfValueTF{#3}
        {\tr_{#3}}
        {\tr}%
    \IfBooleanTF{#1}
        {\paren*{#4}}
        {
            \IfValueTF{#2}
                {\paren[#2]{#4}}
                {\paren{#4}}%
        }%
}

\providecommand\given{}

\newcommand\SetSymbol[1][]{%
    \nonscript\,#1\vert
    \allowbreak
    \nonscript\,
    \mathopen{}}
\DeclarePairedDelimiterX\Set[1]\{\}{%
    \renewcommand\given{%
        \SetSymbol[\delimsize]}
    \nonscript\,
    #1
    \nonscript\,
}

\ExplSyntaxOn
\makeatletter
\cs_generate_variant:Nn \tl_if_eq:nnT {onT}
\tl_new:N \__scale_new_delimsize_tl
\cs_new:Nn \__scale_make_bigger_tl:N {
    \tl_if_eq:onT {#1} {\@MHempty} { \tl_set:Nn \__scale_new_delimsize_tl { big } }
    \tl_if_eq:onT {#1} {}          { \tl_set:Nn \__scale_new_delimsize_tl { big } }
    \tl_if_eq:onT {#1} {\big}      { \tl_set:Nn \__scale_new_delimsize_tl { Big } }
    \tl_if_eq:onT {#1} {\Big}      { \tl_set:Nn \__scale_new_delimsize_tl { bigg } }
    \tl_if_eq:onT {#1} {\bigg}     { \tl_set:Nn \__scale_new_delimsize_tl { Bigg } }
    \tl_if_eq:onT {#1} {\middle}   { \tl_set:Nn \__scale_new_delimsize_tl { middle } }
}
\cs_new:Nn \scale_make_bigger_m:N {
    \__scale_make_bigger_tl:N #1
    \use:c { \tl_use:N \__scale_new_delimsize_tl }
}
\cs_new:Nn \scale_make_bigger_l:N {
    \__scale_make_bigger_tl:N #1
    \tl_if_eq:NnTF \__scale_new_delimsize_tl {middle} {
        \tl_set:Nn \__scale_new_delimsize_tl {left}
    }{
        \tl_put_right:Nn \__scale_new_delimsize_tl {l}
    }
    \use:c { \tl_use:N \__scale_new_delimsize_tl }
}
\cs_new:Nn \scale_make_bigger_r:N {
    \__scale_make_bigger_tl:N #1
    \tl_if_eq:NnTF \__scale_new_delimsize_tl {middle} {
        \tl_set:Nn \__scale_new_delimsize_tl {right}
    }{
        \tl_put_right:Nn \__scale_new_delimsize_tl {r}
    }
    \use:c { \tl_use:N \__scale_new_delimsize_tl }
}

\DeclarePairedDelimiterXPP{\pdist}[1]{\scale_make_bigger_l:N\delimsize\lparen \LdistSYM}{\lparen}{\rparen}{\scale_make_bigger_r:N\delimsize\rparen}{#1}
\DeclarePairedDelimiterXPP{\pdiam}[1]{\scale_make_bigger_l:N\delimsize\lparen \LdiamSYM}{\lparen}{\rparen}{\scale_make_bigger_r:N\delimsize\rparen}{#1}

\makeatother
\ExplSyntaxOff

\usepackage{booktabs}

\usepackage[
    hidelinks,
    bookmarksnumbered,
    bookmarksopen,
    bookmarksopenlevel=3,
    bookmarksdepth=3,
    pdfusetitle=true,
    breaklinks=true,
]{hyperref}
\usepackage{bookmark}
\bookmarksetup{open,numbered}
\usepackage{orcidlink}

\usepackage{caption}
\usepackage{floatrow}
\usepackage{xurl}
\usepackage[
    bibstyle=phys,
    citestyle=numeric,
    biblabel=brackets,
    mincitenames=1,
    maxcitenames=3,
    minalphanames=3,
    maxalphanames=4,
    minbibnames=3,
    maxbibnames=10,
    arxiv=abs,
    eprint=true,
    doi=false,
    url=false,
    sorting=nty,
    useprefix=true,
    alldates=year,
    urldate=long,
]{biblatex}
\renewbibmacro*{note+pages}{%
    \printfield{note}%
    \setunit{\bibpagespunct}%
    \printfield{pages}%
    \newunit%
}
\makeatletter
\DeclareFieldFormat{eprint:arxiv}{%
    \ifhyperref
    {\href{https://arxiv.org/\abx@arxivpath/#1}{%
            arXiv\addcolon\linebreak[2]#1}}
    {arXiv\addcolon
        \nolinkurl{#1}%
    }
}
\makeatother

\renewbibmacro*{maintitle+title}{%
    \iffieldsequal{maintitle}{title}{%
        \clearfield{maintitle}%
        \clearfield{mainsubtitle}%
        \clearfield{maintitleaddon}%
    }{%
        \iffieldundef{maintitle}{}{%
            \printtext[doi/url-link]{%
                \usebibmacro{maintitle}%
            }%
            \newunit\newblock
            \iffieldundef{volume}{}{%
                \printfield{volume}%
                \printfield{part}%
                \setunit{\addcolon\space}%
            }%
        }%
    }%
    \printtext[doi/url-link]{%
        \usebibmacro{title}%
    }%
    \newunit%
}

\DeclareBibliographyDriver{article}{%
    \usebibmacro{bibindex}%
    \usebibmacro{begentry}%
    \usebibmacro{author/translator+others}%
    \setunit{\labelnamepunct}\newblock
    \usebibmacro{title}%
    \newunit
    \printlist{language}%
    \newunit\newblock
    \usebibmacro{byauthor}%
    \newunit\newblock
    \usebibmacro{bytranslator+others}%
    \newunit\newblock
    \printfield{version}%
    \newunit\newblock
    \printtext[doi/url-link]{%
        \usebibmacro{journal+issuetitle}%
        \newunit
        \usebibmacro{byeditor+others}%
        \newunit
        \usebibmacro{note+pages}%
        \newunit\newblock
        \iftoggle{bbx:isbn}
        {\printfield{issn}}
        {}%
    }%
    \newunit\newblock
    \printfield{pubstate}%
    \setunit{\addspace}%
    \printfield{year}%
    \newunit\newblock
    \usebibmacro{doi+eprint+url}%
    \newunit\newblock
    \printfield{addendum}%
    \setunit{\bibpagerefpunct}\newblock
    \usebibmacro{pageref}%
    \newunit\newblock
    \iftoggle{bbx:related}
        {\usebibmacro{related:init}%
        \usebibmacro{related}}
        {}%
    \usebibmacro{finentry}%
}

\makeatletter

\usepackage{amsthm}
\usepackage{thmtools}

\ifdefined\newcounteralias
    \renewcommand\thmt@autorefsetup{%
        \@xa\def\csname\thmt@envname autorefname\@xa\endcsname\@xa{\thmt@thmname}%
    }%
\fi

\def\assumptionthm@envname{assumptionthm}

\newcommand{\redefine@Htheorem@counters}{
    \expandafter\renewcommand\csname theH\thmt@envname\endcsname{\theHsection.\arabic{\thmt@envname}}%
}
\appto{\thmt@newtheorem@postdefinition}{\redefine@Htheorem@counters}

\declaretheorem[
    title=Proposition,%
    parent=section,%
    style=plain,%
]{prop}
\declaretheorem[
    title=Theorem,%
    sibling=prop,%
    style=plain,%
]{theorem}
\declaretheorem[
    title=Lemma,%
    sibling=prop,%
    style=plain,%
]{lem}

\declaretheorem[
    title=Remark,%
    sibling=prop,%
    style=definition,%
]{rmk}

\renewcommand{\redefine@Htheorem@counters}{}

\def\thmt@headstyle@assumption{%
    \ABBREVIATION\ \LONGNAME%
}

\declaretheoremstyle[
    headfont=\normalfont\bfseries,%
    notefont=\normalfont\bfseries,%
    notebraces={}{},%
    headpunct={},%
    headformat=assumption,%
]{assumption}

\declaretheorem[
    name=Assupmtion,%
    numbered=no,%
    style=assumption,%
    qed=\(\diamond\),%
]{assumptionthm}

\newenvironment{assumption}[2]{%
    \def\LONGNAME{#2}%
    \def\ABBREVIATION{#1}%
    \begin{assumptionthm}[%
        name=\upshape #1,%
    ]%
}{
    \end{assumptionthm}
}

\makeatother

\usepackage{tikz}
\usetikzlibrary{calc,cd}
\usetikzlibrary{arrows.meta}
\newlength{\lw}
\newlength{\ul}
\tikzset{%
    myLine/.style = {
        line width = #1\lw,
    },
    myLine/.default = 1,
    x = \ul,
    y = \ul,
    every layer/.style={
        myLine,
        text=black,
    },
    every picture/.style={
        x=\ul,
        y=\ul,
        every layer,
    },
    text = black,
    >={latex},
    shorten/.style = {
        shorten < = #1,
        shorten > = #1
    },
    shorten/.default = 2\lw,
}

\newcommand{\epsi}{\varepsilon}
\newcommand{\E}{{\mathrm{e}}}
\newcommand{\e}{\E}
\newcommand{\I}{\mathrm{i}}

\newcommand{\N}{ \mathbb{N} }
\newcommand{\Z}{ \mathbb{Z} }

\newcommand{\R}{ \mathbb{R} }
\newcommand{\C}{ \mathbb{C} }

\newcommand{\HSpace}{\mathcal{H}}
\newcommand{\calI}{\mathcal{I}}
\newcommand{\calL}{\mathcal{L}}

\newcommand{\unit}{\mathbf{1}}

\newcommand{\Alg}{\mathcal{A}}
\newcommand{\alg}{\Alg}

\newcommand{\D}{\mathrm{d}}

\newcommand{\interpunct}{\mathop{\cdot}}

\newcommand{\vLR}{v_{\mathrm{LR}}}

\newcommand{\fCONV}{f_{\mathrm{conv}}}
\newcommand{\zeCONV}{\zeta_{\mathrm{conv}}}
\newcommand{\fLI}{f_{\mathrm{LI}}}
\newcommand{\zeLI}{\zeta_{\mathrm{LI}}}

\newcommand{\Cvol}{C_{\mathup{vol}}}
\newcommand{\Cint}{C_{\mathup{int}}}
\newcommand{\Cswitch}{C_{\mathup{switch}}}

\newcommand{\Lambdag}{{\Lambda^{\mathup{gap}}}}
\newcommand{\Lambdap}{{\Lambda^{\mathup{pert}}}}

\newcommand{\LI}{\mathrm{LI}}

\newcommand{\LR}{\mathup{LR}}

\DeclarePairedDelimiterX{\innerp}[2]{\langle}{\rangle}{#1,#2}
\newcommand{\suppvshift}[1]{B_R(#1)}

\title{Response theory at low temperature for weakly interacting quantum spin systems}

\author{
    Joscha Henheik%
    \texorpdfstring{%
        \,\orcidlink{0000-0003-1106-327X}%
        \footnote{
            Mathematical Institute, University of Oxford, Andrew Wiles Building, Radcliffe Observatory Quarter, Woodstock Road, Oxford OX2 6GG, United Kingdom \newline
            \mbox{E-mail}:~joscha.henheik@maths.ox.ac.uk
        }
    }{}%
    \and
    Marius Lemm%
    \texorpdfstring{%
        \,\orcidlink{0000-0001-6459-8046}%
        \footnote{
            Fachbereich Mathematik,
            Universität Tübingen,
            Auf der Morgenstelle 10,
            72076 Tübingen,
            Germany.\newline
            \mbox{E-mail}:~marius.lemm@uni-tuebingen.de
        }
    }{}
    \and
    Tom Wessel%
    \texorpdfstring{%
        \,\orcidlink{0000-0001-7593-0913}%
        \footnote{
            Department of Mathematics, Technical University of Munich,
            Boltzmannstr.~3,
            85748 Garching,
            Germany\newline
            Munich Center for Quantum Science and Technology,
            80799 München,
            Germany\newline
            \mbox{E-mail}:~tom.wessel@tum.de
        }
    }{}
}

        \date{October 16, 2026}

\begin{document}

\maketitle

\begin{abstract}
    Recent work by Greenblatt, Lange, Marcelli, and Porta~\cite{GLMP24} provides the first rigorous justification of adiabatic theory for interacting quantum systems at \textit{positive} temperature.
    Their analysis treats weakly interacting lattice fermions and combines a rigorous Wick rotation with fermionic cluster expansion techniques.
    Here, we study a class of weakly interacting quantum spin systems and derive response theory at positive temperature, thereby covering non-Fermi systems for the first time.
    Our proof identifies and exploits the locality property of \textit{local indistinguishability} as a key ingredient in carrying out low-temperature adiabatic expansions.
    The fact that these systems satisfy local indistinguishability was recently established jointly with Adhikari~\cite{AHLW25} using cluster expansion methods together with a probabilistic swapping trick.
\end{abstract}

\section{Introduction and main result}
\label{sec:intro}

An important topic in the mathematical analysis of non-equilibrium physics is the response of a quantum system to adiabatic switching, i.e., to slowly turning on a weak perturbation.
Of particular interest is the justification of \textit{Kubo's formula}~\cite{K57} and its higher-order refinements.
The rigorous response theory of this kind is closely connected to adiabatic theory, which was pioneered by Kato~\cite{K50} and generalized by Nenciu~\cite{N80}.
An influential interpretation through the geometry of vector bundles was discovered by Berry~\cite{B84,S83}.
For subsequent developments of the by now classical adiabatic theory, such as space-adiabatic perturbation theory, we refer to the book~\cite{T03} by Teufel.

Beyond the classical framework, two fundamental obstacles arise when studying adiabatic theory for quantum \textit{many-body} systems: First, the relevant perturbations tend to close the spectral gap that was crucial to standard adiabatic theory.
Second, the error terms diverge with growing system size, making them thermodynamically unstable for extended systems.
In the past years, it has been realized that the key to circumvent both of these obstacles is to exploit locality properties of gapped ground states.
Indeed, Hastings' spectral flow~\cite{HW2005,BMNS2012} has allowed to replace the standard Kato construction of the inverse Liouvillian operator by a quasi-local variant.
Locality provides both dynamical robustness even when the gap closes and error estimates that are uniform in volume, hence thermodynamically stable.
These ideas have been used to derive response theory for local observables in gapped ground states and to prove quantization of the Hall conductance at zero temperature in various settings~\cite{H15,GMP17,BDF17,BDF18,MT2019,Teufel2020,HT2020bulk,HT2020unif,HW2022,HW2024,Wesle2025,Teufel2025,Becker2025}.
We also mention that linear response theory can also be studied outside of an adiabatic framework, typically in a weak coupling regime; see e.g.~\cite{SL78,J06}.\\

Recent work~\cite{GLMP24} by Greenblatt, Lange, Marcelli, and Porta then initiated rigorous response theory for \textit{positive temperature} many-body systems.
Now, the initial state is an equilibrium (Gibbs) state at temperature \(T>0\).
This is the physically more realistic situation, but it requires fundamentally new mathematical techniques to treat mixed states instead of gapped pure states.
We also mention the recent companion works~\cite{JPT24a,JPT24b}, which studied structural aspects of adiabatic theory at positive temperature, as well as a recent physics paper concerned with adiabatic theory at positive temperature~\cite{Ilin2021}.

In~\cite{GLMP24}, the authors consider a system of weakly interacting lattice fermions in \(D\) dimensions subjected to a weak and adiabatically slow perturbation.
They derive response theory through a convergent expansion in terms of the small perturbation size.
Their strategy consists of two main steps: First, they use Lieb-Robinson bounds to justify a rigorous version of the Wick rotation which reformulates \(k\)-point dynamical Duhamel correlation functions as \(k\)-point Euclidean correlation functions.
Second, they calculate these Euclidean correlation functions by fermionic cluster expansion techniques, which heavily exploit diagrammatic cancellations that arise from the antisymmetry.
The errors they obtain are small provided that the temperature \(T\) satisfies \(T\ll \eta^{D+2}\), where \(\eta\in \intervaloc{0,1}\) is a small parameter denoting the inverse of the long time scale on which the adiabatically slow switching occurs.\\

\subsection{Main result}
\label{sec:intro-main-result}

In the present paper, we give the first derivation of response theory for non-Fermi quantum many-body systems: On \(\Z^D\), equipped with the usual graph \(\ell^1\)-metric \(d\colon \Z^D \times \Z^D \to \N_0\), denote an arbitrary finite connected subset by \(\Lambda \Subset \Z^D\).
With every site \(x \in \Z^D\), we associate a finite-dimensional Hilbert space \(\HSpace_x := \C^q\) with \(q \ge 2\) and the corresponding space of linear operators \(\Alg_x := \calL(\HSpace_x)\).
For every \(X\subset \Lambda\) we introduce the many-particle Hilbert space on \(X\) as \(\HSpace_X := \bigotimes_{x \in X} \HSpace_x\) and denote the associated \(C^*\)-algebra of observables by \(\Alg_{X} := \calL(\HSpace_X)\).
Note that, whenever \(X\subset X'\), we have that \(\Alg_{X}\) is naturally embedded as a subalgebra of \(\Alg_{X'}\).

We consider self-adjoint Hamiltonians as the sum of an interacting and non-interacting part,
\begin{equation}
    \label{eq:H0def}
    H_0
    =
    \sum_{x\in \Lambda} h_x
    + \sumstack[lr]{x \in \Lambda : \\ B_R(x) \subset \Lambda} j_x
    .
\end{equation}
Here, in the non-interacting part, every \(h_x \ge 0\) acts only on the site \(x \in \Lambda\), \(h_x \in \Alg_{\List{x}}\), and has a unique zero-energy ground state separated from the rest of the spectrum by a gap of size at least one.
In the interacting part, every \(j_x\) is a small finite-range interaction supported on \(B_R(x)\), a ball of radius \(R\) around site \(x\), i.e.~\(j_x \in \Alg_{B_R(x)}\), and we denote \(\norm{h}_\infty = \sup_{x \in \Z^D} \norm{h_x}_\infty\) and \(\norm{j}_\infty = \sup_{x \in \Z^D} \norm{j_x}_\infty\).
Additionally, we assume that every \(j_x\) is relatively form bounded with respect to the sum of all \(h_x\) in \(B_R(x)\) in the following sense: There exists a small \(a \in \intervaloo{0,1}\) such that
\begin{equation}
    \label{eq:formbound}
    \abs{\innerp{\psi}{j_x \, \psi}}
    \leq
    \frac{a}{\kappa_R} \, \abs[\Big]{\innerp[\Big]{\psi}{\sumstack[lr]{z\in \suppvshift{x}} h_z \, \psi}}
    \qquad\text{for all states \(\psi\) and } \kappa_R=\sup_{x\in \Lambda}\abs{B_R(x)}
\end{equation}
with \(B_r(x):= \Set{y \in \Lambda \given \dist{x,y} \le r}\).
Examples of Hamiltonians that we cover are Heisenberg \(XXZ\) Hamiltonians on any \(D\)-dimensional lattice, subjected to a strong external field in the \(Z\)-direction (which then plays the role of \(h_x\)).
In particular, \eqref{eq:formbound} implies that the product ground state of \(\sum_x h_x\) remains a ground state of \(H_0\).

Starting from the equilibrium state
\begin{equation*}
    \rho_\beta
    =
    \frac{\E^{-\beta H_0}}{\trace{\E^{-\beta H_0}}}
\end{equation*}
as the initial state, we study the response of \(H_0\) from~\eqref{eq:H0def} to a small static perturbation \(\epsi V\) which is slowly (\emph{adiabatically}) turned on.
That is, given a small parameter \(\epsi>0\) and a smooth switching function \(f\colon \R \to \intervalcc{0,1}\) with \(f(t) = 0\) for \(t \le - 1\) and \(f(t) = 1\) for \(t \ge 0\), we consider the time-dependent Hamiltonian
\begin{equation}
    \label{eq:H(t)intro}
    H_{\epsi}(t)
    =
    H_0 + \epsi \, f(t) \, V
    ,\qquad
    V
    =
    \sumstack[lr]{x \in \Lambda : \\ B_R(x) \subset \Lambda} v_x
    ,
\end{equation}
where \(v_x \in \Alg_{B_R(x)}\) is self-adjoint.

To study adiabatic time scales, we fix another small parameter \(\eta>0\).
We write \(\rho_\beta^{\epsi, \eta, f}(t)\) for the solution of the time-dependent adiabatic Schrödinger equation
\begin{equation}
    \label{eq:Schrödinger-equation-switchingintro}
    \I \, \eta \, \frac{\D }{\D t} \rho^{\epsi, \eta, f}_\beta(t)
    =
    \commutator[\big]{H_\epsi(t) , \rho_\beta^{\epsi, \eta, f}(t)}
\end{equation}
and initial datum \(\rho_\beta^{\epsi, \eta, f}(t) = \rho_\beta\) for all \(t \le -1\), where \(\rho_{\beta}\) is the Gibbs state of \(H_0\) at inverse temperature \(\beta \in \intervaloo{0,\infty}\).
For an observable \(A \in \alg_X\), the response to the perturbation is now given by
\begin{equation*}
    \Sigma_{\beta , A}^{\epsi, \eta, f}(t)
    :=
    \trace[\big]{A \, \rho_\beta^{\epsi, \eta, f}(t)}
    - \trace[\big]{A \, \rho_\beta}
    .
\end{equation*}

Our main result, Theorem~\ref{thm:main} below, establishes an expansion of \(\Sigma_{\beta , A}^{\epsi, \eta, f}(t)\) to any order in \(\epsi\) with an additional temperature dependent error term, i.e.~\emph{response theory to any order for low temperature Gibbs states, uniformly in the size of the system}.

\begin{theorem}[Response theory at low temperature]
    \label{thm:main}
    Let \(D\),~\(q\),~\(R \in \N\), and \(\Cint\),~\(\Cswitch>0\).
    Then there exist constants~\(\vLR > 0\) and \(a \in \intervaloo{0,1}\) such that the following holds:
    For all \(n\),~\(m\in \N\) and \(\alpha\in \intervaloo{0,1}\) there exist constants \(C\),~\(c>0\) such that for all \(\Lambda \Subset \Z^D\) and Hamiltonians \(H_0\),~\(V\) as defined in~\eqref{eq:H0def} with \(\norm{h}_\infty\),~\(\norm{j}_\infty\),~\(\norm{v}_\infty \le \Cint\), \(j_x\) and \(v_x\) are of range \(R \in \N\), and \(j_x\) is relatively \(a\)-bounded w.r.t.~the \(h_x\)'s as in~\eqref{eq:formbound}, the following holds:

    Take any inverse temperature \(\beta \in \intervaloo{0,\infty}\), let \(X\subset \Lambda\) be connected and take \(A\in \alg_X\).
    Then there exist response coefficients \(\sigma_{\beta,j}^A\) in the following sense:
    For \(\varepsilon>0\) and smooth switching function \(f\colon \R \to \intervalcc{0,1}\) satisfying \(f(t) = 0\) for \(t \le -1\), \(f(t) = 1\) for \(t \ge 0\), and \(\norm{f}_{C^{\mathfrak{C}_{n,m}}(\R)}<\Cswitch\) with \(\mathfrak{C}_{n,m} := \ceil{m(n + 2D+1)}\), consider the time-dependent Hamiltonian~\eqref{eq:H(t)intro}.
    Let \(\rho_\beta^{\epsi, \eta, f}(t)\) be the solution of the time-dependent adiabatic Schrödinger equation~\eqref{eq:Schrödinger-equation-switchingintro} with adiabatic parameter \(\eta \in \intervaloc{0,1}\) and initial datum \(\rho_\beta^{\epsi, \eta, f}(t) = \rho_\beta\) for all \(t \le -1\), where \(\rho_{\beta}\) is the Gibbs state of \(H_0\) at inverse temperature \(\beta \in \intervaloo{0,\infty}\).

    Then, the response to the perturbation satisfies
    \begin{equation}
        \label{eq:linearresponseintro}
        \begin{aligned}
            \Alignindent
            \abs[\bigg]{\Sigma_{\beta, A}^{\epsi, \eta, f}(t) - \sum_{j=1}^{n} \epsi^j \, \sigma_{\beta,j}^A}
            \\&\le
            C \, \norm{A} \, \E^{\abs{X}^{2/\alpha}} \, \epsi^{n+1} \, \paren[\Big]{\paren[\big]{1 + t}^{2D + 1} + \E^{\vLR(1+t)/\eta - (n+1) \log(\epsi)} \, \E^{-c \beta^{(1 - \alpha)/D}}}
            ,
        \end{aligned}
    \end{equation}
    for any \(\alpha \in \intervaloo{0,1}\), \(\epsi\in \intervaloo{0,1}\), all \(t\geq 0\), and uniformly for \(\eta \in \intervalcc[\big]{\epsi^m , \epsi^{1/m}}\).
    The first order coefficient is given by \emph{Kubo's formula}, see~\eqref{eq:Kubo} below.
\end{theorem}

For \(\beta = \infty\), our main result recovers the usual well known response theory for the unique ground state of \(H_0\), see e.g.~\cite[Theorem~4.1]{Teufel2020}.
Moreover, it shows that, whenever the inverse temperature \(\beta = 1/T\) satisfies
\begin{equation}
    \label{eq:templowerbound}
    \beta
    \ge
    \paren[\bigg]{\frac{\vLR(1+t)}{c\eta} - \frac{(n+1)}{c} \log(\epsi)}^{\frac{D}{1 - \alpha}}
\end{equation}
we effectively have response theory up to order \(n\).
This should be compared to the single previous rigorous work on positive temperature response theory for extended systems~\cite{GLMP24}, which relied on the condition
\begin{equation*}
    \beta \gg \frac{1}{\eta^{D+2}}
\end{equation*}
for \emph{linear} response theory to be valid.
Taking \(\alpha \ll 1\), our bound~\eqref{eq:templowerbound} is slightly stronger, i.e.~we can allow larger temperatures.
Moreover, it is obtained by completely different methods and holds for quantum spin systems instead of lattice Fermions which were considered in~\cite{GLMP24}.
Furthermore, our bound drastically improves if the Hamiltonian~\eqref{eq:H(t)intro} has slow transport properties, as we now explain.

\begin{rmk}[Effect of slower transport]
    \label{rmk:slow}
    The modular nature of the proof immediately implies an improved final result if a stronger Lieb-Robinson bound~\cite{LR1972} holds.
    More precisely, whenever the dynamics \(\tau_{t,s}\) generated by the Hamiltonian~\eqref{eq:H(t)intro} satisfies a Lieb-Robinson bound of the form
    \begin{equation*}
        \norm[\big]{\commutator{\tau_{t,s}(A), B}}
        \le
        C \, \norm{A} \, \norm{B} \, \min\List{\abs{X},\abs{Y}} \, f_{\LR}(\abs{t-s}) \, \E^{-\dist{X,Y}}
    \end{equation*}
    for observables \(A \in \alg_X\) and \(B \in \alg_Y\), we can make the replacement
    \begin{equation*}
        \E^{\vLR(1+t)/\eta} \to f_{\LR}\paren[\big]{(1+t)/\eta}
    \end{equation*}
    in the second line of~\eqref{eq:linearresponseintro}.

    As an example, if \(H_0\) is \emph{locally commuting} then the LR velocity is bounded by \(\sqrt{\epsi}\)~\cite{PHKM10} and so
    \begin{equation*}
        \norm[\big]{\commutator{\tau_{t,s}(A), B}}
        \le
        C \, \norm{A} \, \norm{B} \, \min\List{\abs{X},\abs{Y}} \, \E^{C \, \epsi^{1/2} \, \abs{t-s}} \, \E^{-\dist{X,Y}}
        .
    \end{equation*}
    In this case, the lower bound from~\eqref{eq:templowerbound} drastically improves to
    \begin{equation*}
        \beta
        \ge
        \paren[\bigg]{\frac{C(1+t)}{c} \frac{\epsi^{1/2}}{\eta} - \frac{(n+1)}{c} \log(\epsi)}^{\frac{D}{1 - \alpha}}
        .
    \end{equation*}
    In particular, for \(\epsi = \eta^2\) and neglecting time, we only need \(\beta \gg \abs{\log(\eta)}^{\frac{D}{1 - \alpha}}\).
    Another class of systems where substantially improved Lieb-Robinson bounds can be established are strongly disordered quantum spin systems~\cite{baldwin2023disordered,de2024absence,mcdonough2025non,elgart2026many} and long-range commuting interactions~\cite{LW25}.
\end{rmk}

\subsection{About the proof}
Our proof technique is completely different from that of~\cite{GLMP24}; we do not use the Wick rotation and fermionic cluster expansion techniques are unavailable here.
Instead, we identify and exploit two key drivers to connect low-temperature adiabatic theory to ground state expansions in a way that is uniform in the system size:
\begin{itemize}
    \item \textit{local indistinguishability (LI)}, uniformly in the temperature
    \item \textit{controlled low-energy density of states (CLEDOS)}.
\end{itemize}
Local indistinguishability is a locality concept for Gibbs states originating in quantum information theory~\cite{KGKRE14}, where it was dubbed \emph{locality of temperature}.
It is closely related to exponential decay of correlations~\cite{CMTW25} and establishing it at low temperature is a substantial task in its own right, because it essentially amounts to excluding phase transitions at low temperature.
Controlled low-energy density of states was identified as an important condition by Hastings in~\cite{H07} and further discussed in~\cite{M09}.
Roughly speaking, CLEDOS says that the density of states of low-lying energy intervals is at most polynomial in system size.
CLEDOS has to our knowledge only been verified for non-interacting systems, but is widely expected to hold much more generally for gapped systems, though, there are also gapped Hamiltonians known in which an exponential-in-system-size ground state degeneracy occurs for open boundary conditions~\cite{H07,M09}.

Our proof is then structured in an abstract, modular fashion as displayed in Figure~\ref{fig:overview}.
\begin{figure}
    \centering
    \begin{tikzpicture}[
            node/.style={
                draw,
                rectangle,
                rounded corners,
                minimum width=2.8cm,
                minimum height=1cm,
                inner sep=4\lw,
                align=center,
                node font=\footnotesize,
            },
            arrow/.style={->, thick, shorten},
            backarrow/.style={->, thick, dashed, shorten}
        ]

        \begin{scope}[every node/.append style={node}]
            \node (LI) at (0,2) {
                local indist. \textbf{\nameref{ass:LI}}
                \\ \(\rho_\beta^{\Lambda} \approx \rho_\beta^{\Lambda'}\)
            };
            \node[right] (CONV) at ([xshift=2\ul]LI.east) {
                Gibbs state conv. \textbf{\nameref{ass:conv}}
                \\ \(\rho_\beta^{\Lambda} \xrightarrow{\beta \to \infty} \rho_\infty^{\Lambda}\)
            };
            \node[right] (response) at ([xshift=2\ul]CONV.east) {
                response theory for
                \\ Gibbs states with \(\beta \to \infty\)
            };

            \node[right] (CLEDOS) at (LI.west |- 0,0) {
                controlled low-energy
                \\ DOS \textbf{\nameref{ass:CLEDOS}}
            };
            \node (UGGS) at (CONV |- 0,0) {
                unique gapped
                \\ ground state \textbf{\nameref{ass:gapGS}}
            };
        \end{scope}

        \draw[arrow] (LI) -- (CONV) node[midway, below] (LI to CONV) {\small Prop.~\ref{prop:niceCONV}};
        \draw[arrow] (CONV) -- (response) node[midway, below] (CONV to response) {\small Thm.~\ref{thm:response-theory-all-orders}};

        \begin{scope}[backarrow]
            \draw (CLEDOS) to[bend right=20] (LI to CONV);
            \draw (UGGS) to[bend left=20] (LI to CONV);
            \draw (UGGS) to[bend right=20] (CONV to response);
        \end{scope}
    \end{tikzpicture}
    \caption{
        Structure of the argument.
        Taking \nameref{ass:LI} for the Hamiltonian~\eqref{eq:H0def} as input, in Theorem~\ref{thm:Yarconv} we verify Assumptions~\nameref{ass:conv} and~\nameref{ass:gapGS} based on Proposition~\ref{prop:niceCONV}.
    }
    \label{fig:overview}
\end{figure}
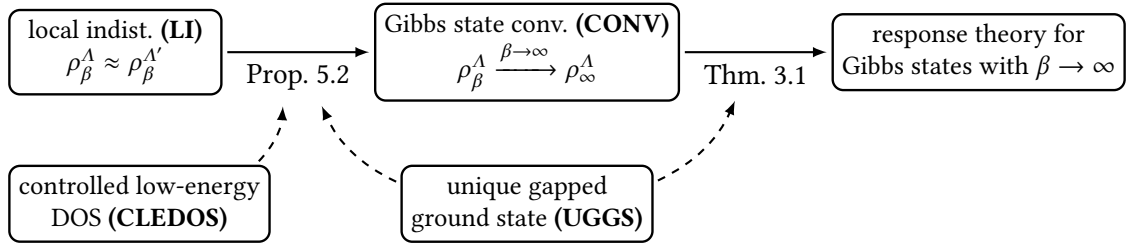
Roughly speaking, local indistinguishability and the controlled low-energy spectrum enable us to approximate the Gibbs state by a local Gibbs state.
Then, the combination of a unique gapped ground state and CLEDOS allows us to approximate the local Gibbs state by the local ground state, the idea being that the gap quantifies the energetic cost and CLEDOS the degeneracy of excited states.
For the local ground state, we can then apply zero-temperature adiabatic theory uniformly in the system size, as done e.g.~in~\cite{MT2019, Teufel2020, HT2020bulk, HW2022}.
A side aspect of the modular approach is that it is easy to track how improvements of different ingredients affect the main result.
For example, stronger Lieb-Robinson bounds (e.g.~for strongly disordered systems, which are covered by our analysis) directly improve the scaling of \(\eta\) with temperature, cf.\ Remark~\ref{rmk:slow}.

Importantly, the class of Hamiltonians \(H_0\) introduced in~\eqref{eq:H0def} provides examples where the modular assumptions, LI and CLEDOS, can be verified.
Indeed, the proof that LI holds uniformly in the temperature was given in our recent joint work~\cite{AHLW25} with Adhikari.
The fact that CLEDOS holds for~\eqref{eq:H0def} is an easy consequence of the variational principle, but it does not seem to have been observed before that CLEDOS is stable under such perturbations and thus it provides a wider class of Hamiltonians where CLEDOS is known compared to~\cite{H07,M09}.

\section{Mathematical setup}
\label{subsec:setup}

In this section, we briefly introduce the (standard) mathematical framework for locally interacting quantum spin systems (see e.g.~\cite{Teufel2020,HT2020unif,HT2020bulk,BRF2018, HW2024} for similar setups), which generalizes the one used in the introduction.

\subsection{Spatial structure and algebra of observables}

We consider a quantum spin system on a finite connected graph~\(\Lambda\) equipped with the usual graph distance \(\dist{\interpunct, \interpunct}\).
Denote the ball of radius \(r\) centred at \(x \in \Lambda\) by \(B_r(x):= \Set{y \in \Lambda \given \dist{x,y} \le r}\).
The graph \(\Lambda\) is assumed to have dimension (at most) \(D >0\), in the sense that there exists a constant \(\Cvol > 0\) such that
\begin{equation}
    \label{eq:dimension}
    \sup_{x \in \Lambda} \abs{B_r(x)} \le 1 + \Cvol \, r^D
    ,
\end{equation}
uniformly in \(r\), where \(\abs{X}\) denotes the number of sites in \(X \subset \Lambda\).
The set of all such graphs \(\Lambda\) is denoted by
\begin{equation*}
    \mathcal{G}(D, \Cvol)
    :=
    \Set[\Big]{
        \Lambda \text{ finite graph}
        \given
        \sup_{x \in \Lambda} \, \abs{B_r(x)} \le 1 + \Cvol \, r^D
        \text{ for all } r > 0
    }
    .
\end{equation*}

To each vertex \(x \in \Lambda\), we associate a single-particle Hilbert space \(\HSpace_x\), which we assume to be of uniformly bounded dimension,
\begin{equation}
    \label{eq:ddef}
    q := \sup_{x \in \Lambda} \dim \HSpace_x < \infty
    .
\end{equation}
We then denote the many-particle Hilbert space \(\HSpace_\Lambda = \bigotimes_{x\in \Lambda} \HSpace_x\) and the C*-algebra of observables \(\Alg_\Lambda := \calL(\HSpace_\Lambda)\).
For every \(X\subset \Lambda\) there is a subalgebra \(\Alg_X\subset \Alg_\Lambda\) isomorphic to \(\calL(\HSpace_X)\).

Note that we will obtain bounds uniform in the graph \(\Lambda\in \mathcal{G}(D, \Cvol)\), but unlike the more common setups, we do not consider a fixed infinite volume~\(\Gamma\) and subsets \(\Lambda \Subset \Gamma\).
Instead, our setup naturally covers torus and cylinder geometries directly.

Since a very similar construction is common for fermionic lattice systems (see e.g.~\cite{NSY2017, HW2022}), all the results almost immediately translate to this setting.

\subsection{Interactions and SLT operators}
\label{subsubsec:interactions}
An \emph{interaction} is a map
\begin{equation*}
    \Phi \colon \Set{X \subset \Lambda} \to \Alg_\Lambda , \ X \mapsto \Phi(X)
    =
    \Phi(X)^*\in \Alg_X
    .
\end{equation*}
Given an interaction one associates a (family of) \emph{sum of local terms (SLT) operators} or \emph{local Hamiltonians} via
\begin{equation}
    \label{eq:intSLT}
    H^{\Lambda'} := \sum_{X \subset \Lambda'} \Phi(X) \in \Alg_\Lambda
    \quadtext{for every} \Lambda' \subset \Lambda
    .
\end{equation}
We will often use the shorthand notation \(H = H^\Lambda\) and call \(H\) and all \(H^{\Lambda'}\) of the form~\eqref{eq:intSLT} SLT operators.
Note that, while every interaction defines a unique operator, there are multiple interactions realizing the same operator, i.e.~the assignment \(\Phi \mapsto H(\Phi)\) is not invertible.
Note that, by definition, all interactions and SLT operators are always self-adjoint.

For every \(b > 0\) and \(p \in \intervaloc{0,1}\), we consider the stretched exponential function
\begin{equation}
    \label{eq:stretched-exponential}
    \chi_{b,p}\colon \intervalco{0,\infty} \to \intervaloc{0,1} ,
    \ x \mapsto \E^{-bx^p}
\end{equation}
as \emph{localization functions} and define the associated \emph{SLT interaction norm}
\begin{equation}
    \label{eq:interaction-norm}
    \norm{\Phi}_{b,p}
    :=
    \sup_{z \in \Lambda} \sumstack{Z\subset \Lambda:\\z\in Z}
    \frac{\norm{\Phi(Z)}}{\chi_{b,p}\pdiam{Z}}
\end{equation}
measuring the quality of localization for an interaction \(\Phi\).
Since we consider only finite graphs, \(\norm{\Phi}_{b,p} < \infty\).
Crucially, our results do not depend on the interaction~\(\Phi\), but only on \(\norm{\Phi}_{b,p}\).
Operators with this property will often be referred to as \emph{\((b,p)\)-localized SLT operators}.

\subsection{Gibbs states}
For any Hilbert space \(\HSpace\) and self-adjoint operator \(H \in \calL(\HSpace)\), we denote the Gibbs state at inverse temperature \(\beta \in \intervaloo{0,\infty}\) by
\begin{equation}
    \label{eq:Gibbsstategeneral}
    \rho_\beta^{\HSpace}[H]
    :=
    \frac{\E^{-\beta H}}{\tr \E^{-\beta H}}
    ,
\end{equation}
where \(\trace{\cdot}\) is the full (non-normalized) trace on \(\HSpace\).
For a system on \(\Lambda \in \mathcal{G}(D, \Cvol)\) and \(\Lambda' \subset \Lambda\), we abbreviate
\begin{equation*}
    \rho_\beta^{\Lambda'}[K]
    :=
    \rho_\beta^{\HSpace_{\Lambda'}}[K]
    \quadtext{for}
    K \in \Alg_{\Lambda'}
    \quadtext{and}
    \rho_\beta^{\Lambda'}
    :=
    \rho_\beta^{\Lambda'}[H^{\Lambda'}]
    ,
\end{equation*}
whenever it is clear, which operator \(H^{\Lambda'}\) is meant.
Moreover, if \(\Lambda' = \Lambda\), we will often drop the \(\Lambda\) superscript entirely, i.e.~write \(\rho_\beta := \rho_\beta^\Lambda\).
For \(\beta=\infty\), we denote with \(\rho_\infty^{\HSpace}(H)\) the normalized ground state projection of~\(H\) and likewise use the above abbreviations.

\section{Response theory at low temperature}

Our main result, Theorem~\ref{thm:main}, will be derived as a special case from a much more general result, Theorem~\ref{thm:response-theory-all-orders} below, which extends zero temperature response theory to low temperature Gibbs states.
It relies on two main assumptions: A uniform gapped ground state (Assumption~\nameref{ass:gapGS}) and the convergence of Gibbs states as \(\beta \to \infty\) (Assumption~\nameref{ass:conv}).
In Theorem~\ref{thm:Yarconv}, these assumptions will then be shown to be satisfied for the systems considered in Theorem~\ref{thm:main}.

We now formulate our two main assumptions.

\begin{assumption}{(UGGS)}{Unique (uniform) gapped ground state.}
    \label{ass:gapGS}
    We say that the SLT Hamiltonian \(H\) has a unique uniformly gapped ground state with gap size \(g > 0\) and minimal system size \(R_0\) (both independent of \(\Lambda\)) if and only if the following holds: For all connected \(\Lambda' \subset \Lambda\), such that there exists \(x \in \Lambda'\) with \(B_{R_0}(x) \subset \Lambda'\), we have that \(E_0^{\Lambda'} := \inf \operatorname{spec}(H^{\Lambda'})\) is a non-degenerate eigenvalue of \(H^{\Lambda'}\) and \(
        \Distp[\big]{E_0^{\Lambda'}, \operatorname{spec}(H^{\Lambda'}) \setminus \List{E_0^{\Lambda'}}}
        \ge
        g
    \).
    If this holds only for \(\Lambda' = \Lambda\), we say that \(H\) has a unique gapped ground state with gap size \(g > 0\).
\end{assumption}

Assumption~\nameref{ass:gapGS} is standard in works on response theory for the ground state, see e.g.~\cite{BRF2018, MT2019, Teufel2020, HT2020, HW2024}.

\begin{assumption}{(CONV)}{Convergence of Gibbs states as \(\beta \to \infty\).}
    \label{ass:conv}
    We say that an SLT Hamiltonian \(H\) with a unique ground state \(\rho_\infty\) has a \emph{converging low-temperature Gibbs state} with respect to two monotonically increasing/decreasing functions \(\fCONV\),~\(\zeCONV \colon \intervalco{0,\infty} \to \intervalco{0,\infty}\), respectively, such that \(\lim_{\beta \to \infty}\zeCONV(\beta) = 0\) and a threshold \(\beta_0 \in \intervalco{0,\infty}\), if and only if
    \begin{equation}
        \label{eq:conv}
        \abs[\big]{\trace[\big]{{\rho_{\beta}} \, A} - \trace[\big]{{\rho_{\infty}} \, A}}
        \le
        \norm{A} \, \fCONV(\abs{X}) \, \zeCONV(\beta)
    \end{equation}
    for all connected \(X \subset \Lambda\), \(A \in \Alg_X\) and \(\beta > \beta_0\).
\end{assumption}

In a nutshell, Assumption~\nameref{ass:conv} says that, when tested on local observables, the Gibbs state converges to the ground state \emph{uniformly} in the system size.
Moreover, by taking the maximum over \(A \in \Alg_X\) with \(\norm{A} = 1\) on the left-hand side of~\eqref{eq:conv}, it expresses a convergence in a local trace norm.
We stress that the uniformity in system size is crucial: In fact, when allowing a \(\abs{\Lambda}\)-dependent prefactor \(\fCONV\) in~\eqref{eq:conv} and assuming that the ground state of \(H\) is separated by a gap of size \(g\) from the rest of the spectrum, such a bound is certainly always true with an exponential rate \(\zeCONV(\beta) = \E^{-g \beta}\); see~Lemma~\ref{lem:betaconv}.

Moreover, we can easily obtain bounds on the best possible forms of the two functions \(\fCONV\) and \(\zeCONV\) by considering the following non-interacting example.
The simplest non-trivial Hamiltonian having a uniformly gapped ground state with gap size \(g> 0\), namely modelling non-interacting spin-\(1/2\) particles in an external magnetic field, is
\begin{equation*}
    H = \frac{g}{2} \sum_{x \in \Lambda} \sigma^3_x
    ,
\end{equation*}
where \(\sigma_x^3\) is the third Pauli matrix acting only on \(x \in \Lambda\).
For this concrete Hamiltonian, it is straightforward to see that (cf.~also~\eqref{eq:cledosest1} and the proof of Lemma~\ref{lem:betaconv})
\begin{align*}
    \max_{\substack{A \in \Alg_X : \\ \norm{A} = 1}} \abs[\big]{\trace[\big]{{\rho_{\beta}} \, A} - \trace[\big]{{\rho_{\infty}} \, A}}
    &=
    2 \, \paren[\bigg]{1 - \frac{1}{\sum_{j=0}^{\abs{X}} \binom{\abs{X}}{j} \, \E^{-\beta g j}}}
    \\&=
    2 \, \paren[\Big]{1 - \paren[\big]{1 + \E^{-\beta g}}^{- \abs{X}}}
    \\&\in
    \intervalcc[\Big]{
        c \, \min\List[\big]{1, \abs{X} \, \E^{-\beta g}}
        ,
        C \, \abs{X} \, \E^{-\beta g}
    }
\end{align*}
for some constants \(C\),~\(c > 0\).
Hence, in~\eqref{eq:conv}, \(\fCONV(\abs{X})\) must necessarily grow at least linearly, and \(\zeCONV(\beta)\) can decay at most exponentially fast with rate given by the gap size~\(g\).

Assuming \nameref{ass:gapGS} and \nameref{ass:conv}, we have response theory for the low temperature Gibbs state to any order, as the following theorem states.

\begin{theorem}[Response theory to all orders]
    \label{thm:response-theory-all-orders}
    Fix \(n\),~\(m\),~\(D\),~\(q\in\N\), \(0<p'<p\le1\), \(0<b'<1\) and \(b > b'\).
    Fix also \(\Cvol\),~\(\Cint\),~\(g\),~\(\Cswitch>0\).
    Then there exist constants \(C\),~\(\vLR > 0\), independent of \(\Lambda\), such that the following holds for every \(\Lambda\in \mathcal{G}(D,\Cvol)\) and all SLT operators \(H_0\),~\(V\) with interactions satisfying \(\norm{\Phi_{H_0}}_{b,p}\),~\(\norm{\Phi_V}_{b,p}\le\Cint\).

    Assume that \(H_0\) has a unique gapped ground state with gap size \(g > 0\) and satisfies Assumption~\nameref{ass:conv} w.r.t.~the functions \(\fCONV\),~\(\zeCONV\) and threshold \(\beta_0 \in \intervalco{0,\infty}\).
    Take any inverse temperature \(\beta > \beta_0\).
    Let \(X\subset \Lambda\) be connected and take \(A\in \alg_X\).
    Then there exist response coefficients \(\sigma_{\beta,j}^A\) in the following sense:
    For \(\varepsilon\in\intervaloc{0,1}\) and a smooth switching function \(f\colon \R \to \intervalcc{0,1}\) satisfying \(f(t) = 0\) for \(t \le -1\), \(f(t) = 1\) for \(t \ge 0\), and \(\norm{f}_{C^{\mathfrak{C}_{n,m}}(\R)}<\Cswitch\), where
    \(
        \mathfrak{C}_{n,m}:=\ceil{m(n+1+2D/p)}
        ,
    \)
    consider the time-dependent Hamiltonian
    \begin{equation*}
        H_\epsi(t) := H_0 + \epsi \, f(t) \, V
        .
    \end{equation*}
    Let \(\rho_\beta^{\epsi, \eta, f}(t)\) be the solution of the time-dependent adiabatic Schrödinger equation
    \begin{equation*}
        \I \, \eta \, \frac{\D }{\D t} \rho^{\epsi, \eta, f}_\beta(t)
        =
        \commutator[\big]{H_\epsi(t) , \rho_\beta^{\epsi, \eta, f}(t)}
    \end{equation*}
    with adiabatic parameter \(\eta \in \intervaloc{0,1}\) and initial datum \(\rho_\beta^{\epsi, \eta, f}(t) = \rho_\beta\) for all \(t \le -1\), where \(\rho_{\beta}\) is the Gibbs state of \(H_0\) at inverse temperature~\(\beta\).

    Then, the response to the perturbation,
    \(
        \Sigma_{\beta , A}^{\epsi, \eta, f}(t)
        :=
        \trace[\big]{A \, \rho_\beta^{\epsi, \eta, f}(t)}
        - \trace[\big]{A \, \rho_\beta}
    \),
    satisfies
    \begin{equation}
        \label{eq:linearresponse}
        \begin{aligned}
            \Alignindent
            \abs[\bigg]{\Sigma_{\beta, A}^{\epsi, \eta, f}(t) - \sum_{j=1}^{n} \epsi^j \, \sigma_{\beta,j}^A}
            \\&\le
            \begin{aligned}[t]
                &C \, \norm{A} \, \abs{X}^{(3+n)D} \, \paren[\big]{1 + t}^{2D/p + 1} \, \epsi^{n+1}
                \\&+ C \, \norm{A} \, \abs{X}^{n} \, \inf_{r>0} \paren[\Big]{
                    \E^{\vLR(1+t)/\eta} \, \chi_{b',p'}(r)
                    + \fCONV\paren[\big]{\abs{X} \, (1 + \Cvol r^D)} \, \zeCONV(\beta)
                }
                ,
            \end{aligned}
        \end{aligned}
    \end{equation}
    for all \(t\geq 0\), uniformly for \(\eta \in \intervalcc[\big]{\epsi^m , \epsi^{1/m}}\).
    The first order coefficient is given by \emph{Kubo's formula}
    \begin{equation}
        \label{eq:Kubo}
        \sigma_{\beta,1}^A
        =
        - \I \, \trace[\big]{\commutator[\big]{\calI_{H_0,g}(V), A} \, \rho_\beta}
        ,
    \end{equation}
    where \(\calI_{H_0,g}(\interpunct)\) is the quasi-local inverse of the Liouvillian with suitable filter function introduced in Remark~\ref{rmk:explicit-response-coefficients} below.
\end{theorem}

\begin{rmk}[Explicit response coefficients]
    \label{rmk:explicit-response-coefficients}
    The response coefficients from Theorem~\ref{thm:response-theory-all-orders} can be explicitly calculated from the initial state.
    The key object \(\calI_{H_0,g}(\interpunct)\) is the inverse Liouvillian
    \begin{equation*}
        \calI_{H,g}(O)
        :=
        \int_\R \D t \, w_g(t) \int_0^t \D s \, \e^{\I Hs}
        \, O \, \e^{-\I Hs}
        ,
    \end{equation*}
    defined via a \emph{filter function} \(w_g \in L^1(\R)\) that is positive, normalized to \(\int w_g = 1\), and has Fourier transform%
    \footnote{%
        We use the convention that
        \(\widehat{f}(p) := (2 \pi)^{-1/2} \int_\R \D x \, \E^{- \I p x} f(x)\)
        for the Fourier transform.
    }
    \(\widehat{w_g}\) with support
    \begin{equation*}
        \support(\widehat{w_g}) \subset \intervalcc{-g,g}
        ,
    \end{equation*}
    and for all \(p<1\) decays as
    \begin{equation}
        \label{eq:wbound}
        \abs{w_g(t)} \le C \, \e^{-\abs{t}^p}
        .
    \end{equation}
    A function satisfying these requirements can be explicitly constructed
    (see e.g.~\cite{HW2005,BMNS2012}) and is discussed in~\cite[Appendix~A.3]{HW2024} in more detail.
    The higher orders are constructed iteratively; see~\cite[Section~6.2 and Proposition~4.5]{HW2024}.
\end{rmk}

The proof of Theorem~\ref{thm:response-theory-all-orders} relies on the method of constructing \emph{non-equilibrium almost stationary states} (NEASSs) (see e.g.~\cite{MT2019, Teufel2020, HT2020bulk, HT2020unif, HW2022, HW2024, Wesle2025, Becker2025}) and is given in Section~\ref{sec:NEASS}.

\section{Proof of Theorem~\ref{thm:main} via convergence of Gibbs states}

The following theorem shows that the systems considered in Theorem~\ref{thm:main} in fact satisfy Assumptions~\nameref{ass:gapGS} and \nameref{ass:conv}.

\begin{theorem}[Rapid ground state convergence]
    \label{thm:Yarconv}
    Fix \(D\),~\(q\),~\(R \in \N\) and \(\Cint > 0\).
    Then there exist \(a\in \intervaloo{0,1}\) and \(C_1\),~\(C_2\),~\(C_3 > 0\) such that the following holds.
    Consider a lattice~\(\Lambda \Subset \Z^D\) and a Hamiltonian~\(H_0\) as defined in~\eqref{eq:H0def} with \(\norm{h}_\infty\), \(\norm{j}_\infty \le \Cint\), \(j_x\) of range~\(R\), and \(j_x\) relatively \(a\)-bounded w.r.t.~the~\(h_x\)'s in the sense~\eqref{eq:formbound}.

    Then, \(H_0\) has a unique gapped ground state \(\rho_\infty\) with gap~\(g \ge 1 - a\), and we have that its Gibbs state \(\rho_\beta\) converges exponentially fast to the ground state as \(\beta \to \infty\), i.e.\
    \begin{equation*}
        \abs[\big]{
            \trace{A \, \rho_\beta}
            - \trace{A \, \rho_\infty}
        }
        \le
        C_1 \, \norm{A} \, \exp\paren[\big]{C_2 \, \abs{X}} \, \exp\paren[\big]{- C_3 \, \beta}
    \end{equation*}
    for all connected sets \(X\subset \Lambda\) and observables \(A \in \Alg_X\).
\end{theorem}

The proof of Theorem~\ref{thm:Yarconv} is given in Section~\ref{sec:LI}.
\begin{proof}[Proof of Theorem~\ref{thm:main}]
    This is an immediate consequence of Theorems~\ref{thm:response-theory-all-orders} and~\ref{thm:Yarconv}.
    The minimum over \(r\) is estimated by applying the Young inequality
    \(\abs{X} r^D \le \alpha \abs{X}^{1/\alpha} + (1-\alpha) r^{D/(1-\alpha)}\) and choosing \(r \sim \beta^{(1-\alpha)/D}\).
    Afterwards, we adjust the exponent of \(\abs{X}\) to absorb multiplicative prefactors.
\end{proof}

\section{Proof of Theorem~\ref{thm:Yarconv}}
\label{sec:LI}
The goal of this section is to give the proof of Theorem~\ref{thm:Yarconv}.
In order to do so, in Proposition~\ref{prop:niceCONV} below, we will show that Assumption~\nameref{ass:conv} is a consequence of \emph{local indistinguishability}, see Assumption~\nameref{ass:LI}, and a \emph{controlled low-energy density of states}, see Assumption~\nameref{ass:CLEDOS}.

\subsection{Local indistinguishability and controlled low-energy density of states}
We start by formulating the local indistinguishability assumption precisely.

\begin{assumption}{(LI)}{Local indistinguishability.}
    \label{ass:LI}
    We say that an SLT Hamiltonian \(H\) satisfies \emph{local indistinguishability} (LI) at inverse temperature \(\beta \in \intervaloo{0,\infty}\) on \(\Lambda \in \mathcal{G}(D, \Cvol)\) with minimal system size \(R_0\) (independent of \(\Lambda\)) and with respect to two monotonically increasing/decreasing functions \(\fLI\),~\(\zeLI \colon \intervalco{0,\infty} \to \intervalco{0,\infty}\), respectively, with \(\zeLI(r) \to 0\) as \(r \to \infty\), if and only if the following holds: For all connected \(\Lambda' \subset \Lambda\) such that there exists \(x \in \Lambda'\) with \(B_{R_0}(x) \subset \Lambda'\), we have that
    \begin{equation}
        \label{eq:LI}
        \abs[\big]{\trace[\big]{\rho_\beta^{\Lambda} \, A} - \trace[\big]{\rho_\beta^{\Lambda'} \, A}}
        \le
        \norm{A} \, \fLI(\abs{X}) \, \zeLI\pdist{X, \Lambda \setminus \Lambda'}
    \end{equation}
    for all connected \(X \subset \Lambda'\) and \(A \in \Alg_X\).
    Moreover, if there exists \(\beta_0 \in \intervalco{0,\infty}\) such that \(H\) satisfies LI for all \(\beta > \beta_0\) with the same functions \(\fLI\),~\(\zeLI\), we say that \(H\) satisfies \emph{LI uniformly in small temperatures} w.r.t.~\(\fLI\),~\(\zeLI\) and threshold \(\beta_0\).
\end{assumption}

Next, we formulate the assumption of a controlled low-energy density of states precisely: To do so, for a fixed constant \(\delta > 0\), we denote the intervals
\begin{equation*}
    I_k
    =
    I_k(\delta)
    :=
    \intervalco{k\delta, (k+1)\delta}
    \quadtext{for}
    k \in \N_0
    .
\end{equation*}

\begin{assumption}{(CLEDOS)}{Controlled low-energy density of states.}
    \label{ass:CLEDOS}
    Let \(H\) be an SLT Hamiltonian on \(\Lambda \in \mathcal{G}(D, \Cvol)\) with ground state eigenvalue zero, i.e.~\(\inf \operatorname{spec}(H) = 0\).
    We say that \(H\) has a uniform \emph{controlled low-energy density of states} (CLEDOS) with minimal system size \(R_0\), w.r.t.~\(\delta = \delta_{\mathrm{cledos}} > 0\) and \(m = m_{\mathrm{cledos}} \in \N\) if and only if the following holds: For all connected \(\Lambda' \subset \Lambda\) such that there exists \(x \in \Lambda'\) with \(B_{R_0}(x) \subset \Lambda'\), we have that
    \begin{equation}
        \label{eq:cledos}
        \inf \operatorname{spec}\paren[\big]{H^{\Lambda'}}
        =
        0
        \quadtext{and}
        \trace[\big]{\chi_{I_k(\delta)}(H^{\Lambda'})}
        \le
        1 + k^m \abs{\Lambda'}^{mk}
        \quadtext{for all}
        k \in \N_0
        ,
    \end{equation}
    where \(\chi_I\) denotes the characteristic function of an interval \(I \subset \R\).
\end{assumption}

Note that~\eqref{eq:cledos} for \(k=0\) implies that, in particular, zero is a non-degenerate ground state eigenvalue of \(H^{\Lambda'}\), which is separated by a gap of size at least \(\delta\).
More precisely, \nameref{ass:CLEDOS} implies \nameref{ass:gapGS} with \(g = \delta\).
For general \(k\), Assumption~\nameref{ass:CLEDOS} states that the number of eigenvalues in the interval \(I_k(\delta)\) can grow at most polynomially in the system size \(\abs{\Lambda'}\).

Further, we point out that, upon adjusting \(m \to 3m\), the bound~\eqref{eq:cledos} is equivalent to controlling
\begin{equation}
    \label{eq:Itofull}
    \trace[\big]{\chi_{\intervalco{0,k\delta}}(H^{\Lambda'})}
    =
    \sum_{\ell = 0}^{k-1} \trace[\big]{\chi_{I_\ell(\delta)}(H^{\Lambda'})}
    .
\end{equation}

\begin{rmk}[On Assumption~\nameref{ass:CLEDOS}]
    \label{rmk:cledos}
    As we exploit in the proof of Theorem~\ref{thm:Yarconv}, Assumption~\nameref{ass:CLEDOS} holds true for non-interacting Hamiltonians with a uniform on-site gap.
    Additionally, as employed in the proof of Theorem~\ref{thm:Yarconv} as well, relatively form-bounded perturbations of such non-interacting on-site gapped Hamiltonians (with a form bound constant smaller than the minimal on-site gap size) preserve Assumption~\nameref{ass:CLEDOS} upon adjusting the parameters \(\delta\) and~\(m\).
    More generally, we expect Assumption~\nameref{ass:CLEDOS} to be generically true for weak perturbations of on-site gapped Hamiltonians, i.e.~not necessarily satisfying the form bound.
    In fact, the essence of Assumption~\nameref{ass:CLEDOS} is that \(H\) has a unique gapped ground state and at low excited energies, the many-body density of states (mbDoS) is at most polynomially large in the system size.
\end{rmk}

We can now formulate the main result of this section, stating that, under additional spectral conditions, \nameref{ass:LI} implies \nameref{ass:conv}.
The proof is given in Section~\ref{sec:proof-LItoCONV}.

\begin{prop}[\nameref{ass:LI} implies \nameref{ass:conv}]
    \label{prop:niceCONV}
    Let \(H\) be an SLT Hamiltonian that satisfies Assumption~\nameref{ass:LI} with minimal system size \(R_0\), uniformly in small temperatures w.r.t.~\(\fLI\),~\(\zeLI\) and threshold \(\beta_0 \in \intervalco{0,\infty}\).
    Suppose, additionally, that \(H\) has a unique ground state.
    Then we have the following:
    \begin{itemize}
        \item[(i)] \textnormal{[Using \nameref{ass:CLEDOS}]} Assume additionally that \(H\) satisfies Assumption~\nameref{ass:CLEDOS} w.r.t.~\(\delta = \delta_{\mathrm{cledos}}\), \(m = m_{\mathrm{cledos}}\) and the same minimal size~\(R_0\). Then \(H\) has a converging low-temperature Gibbs state w.r.t.~the threshold \(\beta_0\) and the functions
            \begin{align*}
                \fCONV(\abs{X})
                &=
                2 \, \paren[\big]{\fLI(\abs{X}) + C_m \, (1+\abs{X}^m)}
            \shortintertext{and}
                \zeCONV(\beta)
                &=
                \inf_{r \in \intervaloo{R_0,\infty}} \paren[\Big]{\zeLI(r) + C_m \, \paren[\big]{1 + \paren{1 + \Cvol \, r^D}^m} \, \E^{-\delta \beta}}
                ,
            \end{align*}
            where \(C_m\) is a constant depending only on \(m\).
        \item[(ii)] \textnormal{[Using \nameref{ass:gapGS}]} Assume additionally that \(H\) satisfies Assumption~\nameref{ass:gapGS} with parameters \(g\) and \(R_0\). Then \(H\) has a converging low-temperature Gibbs state w.r.t.~the threshold \(\beta_0\) and the functions
            \begin{align*}
                \fCONV(\abs{X})
                &=
                2 \, \paren[\big]{\fLI(\abs{X}) + \E^{\log(q) \alpha \abs{X}^{1/\alpha}}}
            \shortintertext{and}
                \zeCONV(\beta)
                &=
                \inf_{r \in \intervaloo{R_0,\infty}} \paren[\bigg]{\zeLI(r) + \E^{\log(q) (1-\alpha) \paren[\big]{1 + \Cvol r^D}^{\frac{1}{1 - \alpha}}} \, \E^{-g \beta}}
            \end{align*}
            for any \(\alpha \in \intervaloo{0,1}\).
    \end{itemize}
\end{prop}

\subsection{Rapid ground state convergence via LI: Proof of Theorem~\ref{thm:Yarconv}}
We will show Theorem~\ref{thm:Yarconv} based on Proposition~\ref{prop:niceCONV}.
In order to do so, we first show that \(H_0\) from~\eqref{eq:H0def} satisfies Assumption~\nameref{ass:CLEDOS}.
Let \(\Lambda'\subset \Lambda\) and denote the non-interacting part of~\eqref{eq:H0def} by
\begin{equation*}
    H_{\mathrm{non}}^{\Lambda'}
    =
    \sum_{x \in \Lambda'} h_x
    .
\end{equation*}
Then, using the fact that \(H_{\mathrm{non}}^{\Lambda'}\) is non-interacting and each \(h_x \ge 0\) has a gap of size at least one above its ground state eigenvalue zero, it holds that
\begin{align*}
    \trace[\big]{\chi_{\intervalco{0,k}}(H_{\mathrm{non}}^{\Lambda'})}
    &\le
    \sum_{\ell = 0}^{k-1} \binom{\abs{\Lambda'}}{\ell} \, (q - 1)^\ell
    \\&\le
    \sum_{\ell = 0}^{k-1}\frac{1}{\ell!} \, \abs{\Lambda'}^\ell \, (q - 1)^\ell
    \le
    1 + \unit_{k \ge 2} \, \E^{q-1} \, \abs{\Lambda'}^k
    .
\end{align*}
Taking \(m = \ceil{(q-1)/( \log 2)}\), we then have that \(\unit_{k \ge 2} \, \E^{q-1} \le k^m\).

Next, by the form bound~\eqref{eq:formbound}, we have \((1-a) H_{\mathrm{non}}^{\Lambda'} \leq H_0^{\Lambda'} \leq (1+a) H_{\mathrm{non}}^{\Lambda'}\).
Hence, by the min-max principle, we have that the \(j^{\text{th}}\) eigenvalue \(E_j^{0,\Lambda'}\) of \(H_0^{\Lambda'}\) remains \enquote{of the same order} as the \(j^{\mathrm{th}}\) eigenvalue \(E_j^{\mathrm{non},\Lambda'}\) of \(H_{\mathrm{non}}^{\Lambda'}\), i.e.~\(E_j^{0,\Lambda'} \in \intervalcc{(1 - a)E_j^{\mathrm{non},\Lambda'}, (1 + a) E_j^{\mathrm{non},\Lambda'}}\); cf.~also Remark~\ref{rmk:cledos}.
Hence,
\begin{equation*}
    \trace[\big]{\chi_{\intervalco{0,(1 - a)k}}(H_{0}^{\Lambda'})}
    \le
    \sum_{\ell = 0}^{k-1} \binom{\abs{\Lambda'}}{\ell} \, (q - 1)^\ell
    \le
    1 + \unit_{k \ge 2} \, \E^{q-1} \, \abs{\Lambda'}^k
\end{equation*}
and thus, taking \(m = \ceil{(q-1)/ \log 2}\), we have that \(\unit_{k \ge 2} \, \E^{q-1} \le k^m\).

In particular, involving the argument around~\eqref{eq:Itofull}, \(H_{0}\) satisfies Assumption~\nameref{ass:CLEDOS} w.r.t.~the parameters \(\delta = \tfrac{1-a}{2}\) and \(m = 3 \, \ceil{(q-1)/\log 2}\).

The claim now follows from Proposition~\ref{prop:niceCONV}~(i) by involving the fact that \(H_0\) from~\eqref{eq:H0def} satisfies Assumption~\nameref{ass:LI} at every inverse temperature \(\beta \in \intervaloo{0,\infty}\) with
\begin{equation*}
    \fLI(\abs{X})
    =
    C_1 \, \exp(C_2 \, \abs{X})
    \quadtext{and}
    \zeLI\pdist{X, \Lambda \setminus \Lambda'}
    =
    \exp\paren[\big]{- \dist{X, \Lambda \setminus \Lambda'}/\xi_{\LI}}
\end{equation*}
for some constants \(C_1\),~\(C_2\),~\(\xi_{\LI} > 0\), as proven in~\cite[Theorem~2.3]{AHLW25}.
The exponential decay of \(\zeLI\) allows to choose \(r \sim \beta\) in the infimum of \(\zeCONV\) provided in Proposition~\ref{prop:niceCONV} and obtain the exponential decay~\(\e^{-C_3 \, \beta}\) in the statement.

\subsection{From LI to CONV: Proof of Proposition~\ref{prop:niceCONV}}
\label{sec:proof-LItoCONV}

Let us first notice that \nameref{ass:LI} extends to the ground state under the assumptions of the proposition.
Therefore, notice that \nameref{ass:CLEDOS} at \(k=0\) and \nameref{ass:gapGS} each imply that the ground states \(\rho^{\Lambda'}_\infty\) of \(H^{\Lambda'}\) are non-degenerate for the relevant~\(\Lambda'\) allowed in the assumptions.
Hence, \(\rho^{\Lambda'}_\beta \to \rho^{\Lambda'}_\infty\) in trace norm.
And since the LI bound is uniform in~\(\beta>\beta_0\), it extends to the limit.

Now, let \(X \subset \Lambda\) be connected, \(A \in \Alg_X\) and \(\beta > \beta_0\), so that~\eqref{eq:LI} is available with the functions \(\fLI\),~\(\zeLI\).
We begin with the first statement in Proposition~\ref{prop:niceCONV} and estimate
\begin{equation}
    \label{eq:eps/3}
    \begin{aligned}
        \Alignindent
        \abs[\big]{\trace[\big]{{\rho_{\beta}^\Lambda} \, A} - \trace[\big]{{\rho_{\infty}^\Lambda} \, A}}
        \\&\le
        \abs[\big]{\trace[\big]{{\rho_{\beta}^\Lambda} \, A} - \trace[\big]{{\rho_{\beta}^{\Lambda'}} \, A}}
        + \abs[\big]{\trace[\big]{{\rho_{\beta}^{\Lambda'}} \, A} - \trace[\big]{{\rho_{\infty}^{\Lambda'}} \, A}}
        + \abs[\big]{\trace[\big]{{\rho_{\infty}^{\Lambda'}} \, A} - \trace[\big]{{\rho_{\infty}^\Lambda} \, A}}
    \end{aligned}
\end{equation}
for some connected \(\Lambda'\) satisfying \(X \subset \Lambda' \subset \Lambda\) to be chosen below.
Thus the first and third terms in~\eqref{eq:eps/3} satisfy
\begin{equation}
    \label{eq:LIestimate}
    \abs[\big]{\trace[\big]{{\rho_{\beta}^\Lambda} \, A} - \trace[\big]{{\rho_{\beta}^{\Lambda'}} \, A}} + \abs[\big]{\trace[\big]{{\rho_{\infty}^{\Lambda'}} \, A} - \trace[\big]{{\rho_{\infty}^\Lambda} \, A}}
    \le
    2 \, \norm{A} \, \fLI(\abs{X}) \, \zeLI\pdist{X, \Lambda \setminus \Lambda'}
    .
\end{equation}
For the second term on the right-hand side of~\eqref{eq:eps/3}, we employ Assumption~\nameref{ass:CLEDOS}.
In order to do this, let
\begin{equation*}
    Z_\beta^{\Lambda'}
    :=
    \tr \E^{-\beta H^{\Lambda'}}
\end{equation*}
be the partition function of \(H^{\Lambda'}\) at inverse temperature \(\beta\).
Since zero is a non-degenerate ground state eigenvalue of \(H^{\Lambda'}\), we find
\begin{equation}
    \label{eq:cledosest1}
    \abs[\big]{\trace[\big]{{\rho_{\beta}^{\Lambda'}} \, A} - \trace[\big]{{\rho_{\infty}^{\Lambda'}} \, A}}
    \le
    \norm{A} \, \tr\,\abs[\big]{{\rho_{\beta}^{\Lambda'}} - {\rho_{\infty}^{\Lambda'}}}
    =
    2 \, \norm{A} \, \frac{Z_\beta^{\Lambda'} - 1}{Z_\beta^{\Lambda'}}
    ,
\end{equation}
where we used that
\begin{equation*}
    \tr\,\abs[\big]{{\rho_{\beta}^{\Lambda'}} - {\rho_{\infty}^{\Lambda'}}}
    =
    \paren[\bigg]{1 - \frac{1}{Z_\beta^{\Lambda'}}} + \frac{Z_\beta^{\Lambda'}- 1}{Z_\beta^{\Lambda'}}
\end{equation*}
with the two contributions coming from the ground state eigenvalue and the remaining excited eigenvalues, respectively, since \(\inf \operatorname{spec} (H^{\Lambda'}) = 0\).
We are hence left with proving an upper bound on the partition function.

In order to do so, we employ Assumption~\nameref{ass:CLEDOS} and find the partition function to be bounded as
\begin{equation}
    \label{eq:cledosest2}
    Z_\beta^{\Lambda'}
    \le
    1 + \sum_{k=1}^\infty (1 + k^m \abs{\Lambda'}^{mk}) \, \E^{-\beta \delta k}
    \le
    1 + C_m \, (1 + \abs{\Lambda'}^m) \, \E^{-\delta \beta}
\end{equation}
for some constant \(C_m\) depending only on \(m\), if \(\beta\) is large enough such that \((1 + \abs{\Lambda'}^m) \, \E^{-\delta \beta} < 1/2\) (otherwise the following bound~\eqref{eq:convfinal} is anyway trivial).

Hence, combining~\eqref{eq:eps/3} with~\eqref{eq:LIestimate} and~\eqref{eq:cledosest1}--\eqref{eq:cledosest2}, we obtain
\begin{equation}
    \label{eq:convfinal}
    \abs[\big]{\trace[\big]{{\rho_{\beta}^\Lambda} \, A} - \trace[\big]{{\rho_{\infty}^\Lambda} \, A}}
    \le
    2 \, \norm{A} \, \paren[\Big]{\fLI(\abs{X}) \, \zeLI\pdist{X, \Lambda \setminus \Lambda'} + C_m \, (1 + \abs{\Lambda'}^m) \, \E^{-\delta \beta}}
    .
\end{equation}
We now pick \(\Lambda' = X_r = \Set{x \in \Lambda \given \dist{x, X} \le r}\), the \(r\)-fattening of \(X\) in \(\Lambda\) with \(r > R_0\), and find that we can choose
\begin{align*}
    \fCONV(\abs{X})
    &=
    2 \, \paren[\big]{\fLI(\abs{X}) + C_m \, (1+\abs{X}^m)}
\shortintertext{and}
    \zeCONV(\beta)
    &=
    \inf_{r \in \intervaloo{R_0,\infty}} \paren[\Big]{\zeLI(r) + C_m \, \paren[\big]{1 + \paren{1 + \Cvol \, r^D}^m} \, \E^{-\delta \beta}}
    .
\end{align*}
Note that \(\fCONV\) is monotonically increasing and \(\zeCONV\) is monotonically decreasing.

We now turn to proving the second statement in Proposition~\ref{prop:niceCONV}.
Similarly to the first statement, we apply~\eqref{eq:eps/3} and~\eqref{eq:LIestimate}.
However, now for the middle term of~\eqref{eq:eps/3} we employ Lemma~\ref{lem:betaconv} instead of~\eqref{eq:cledosest1}--\eqref{eq:cledosest2}: In fact, we have that
\begin{equation*}
    \abs[\big]{\trace[\big]{{\rho_{\beta}^{\Lambda'}} \, A} - \trace[\big]{{\rho_{\infty}^{\Lambda'}} \, A}}
    \le
    \norm{A} \, \norm{{\rho_{\beta}^{\Lambda'}} - {\rho_{\infty}^{\Lambda'}}}_{\tr}
    \le
    2 \, \norm{A} \, \E^{\log(q) \abs{\Lambda'}} \, \E^{-g \beta}
    ,
\end{equation*}
where \(q\) is the dimension of the local Hilbert spaces~\(\HSpace_{x}\).
Consequently,
\begin{equation*}
    \abs[\big]{\trace[\big]{{\rho_{\beta}^\Lambda} \, A} - \trace[\big]{{\rho_{\infty}^\Lambda} \, A}}
    \le
    2 \, \norm{A} \, \paren[\Big]{\fLI(\abs{X}) \, \zeLI\pdist{X, \Lambda \setminus \Lambda'} + \E^{\log(q) \abs{\Lambda'}} \, \E^{-g \beta}}
    .
\end{equation*}
Picking again \(\Lambda' = X_r = \Set{x \in \Lambda \given \dist{x, X} \le r}\) for \(r > R_0\) and optimizing in \(r\), we find that we can choose, for any \(\alpha \in \intervaloo{0,1}\) by a Young inequality,
\begin{align*}
    \fCONV(\abs{X})
    &=
    2 \, \paren[\Big]{\fLI(\abs{X}) + \E^{\log(q) \alpha \abs{X}^{1/\alpha}}}
\shortintertext{and}
    \zeCONV(\beta)
    &=
    \inf_{r \in \intervaloo{R_0,\infty}} \paren[\bigg]{\zeLI(r) + \E^{\log(q) (1-\alpha) \paren[\big]{1 + \Cvol \, r^D}^{\frac{1}{1 - \alpha}}} \, \E^{-g \beta}}
    .
\end{align*}
This concludes the proof.

\section{NEASS construction: Proof of Theorem~\ref{thm:response-theory-all-orders}}
\label{sec:NEASS}

As mentioned before, the proof of Theorem~\ref{thm:response-theory-all-orders} builds on response theory for the ground state.
Contrary to the positive temperature setting, response theory is well understood for gapped ground states.
More specifically, one can construct so-called \emph{non-equilibrium almost stationary states} (NEASSs), which approximate the state into which the system is driven very well and allow to compute expectation values of observables explicitly~\cite{BRF2018,Teufel2020,HW2022,HW2024}.
We now state the ground-state version of Theorem~\ref{thm:response-theory-all-orders}.

For \(X\subset\Lambda\), let \(\Pi_{X_r}\colon \Alg_\Lambda\to\Alg_{X_r}\) be the normalized-partial-trace conditional expectation, see~\cite[Lemma~4.1]{NSY2019} or~\cite[Lemma~A.6]{HW2024} for details and crucial properties.
It is used to restrict the support of spread out observables in the following statement.

\begin{theorem}[Response theory for gapped ground states ({adaptation of~\cite[Theorem~4.1]{HW2024}})]
    \label{thm:linear-response-ground-states}
    Fix \(n\),~\(m\),~\(D\),~\(q\in\N\), \(0<p'<p\leq 1\), \(0<b'<1\) and \(b > b'\).
    Fix also \(\Cvol\),~\(\Cint\),~\(g\),~\(\Cswitch>0\).
    Then there exist constants \(C\),~\(C_j > 0\) such that the following holds for every \(\Lambda\in \mathcal{G}(D,\Cvol)\) and all SLT operators \(H_0\),~\(V\) with interactions satisfying \(\norm{\Phi_{H_0}}_{b,p}\),~\(\norm{\Phi_V}_{b,p} \leq \Cint\).

    Assume \(H_0\) has a unique gapped ground state \(\rho_\infty\) with gap size \(g>0\).
    Let \(X\subset \Lambda\) be connected and take \(A\in \alg_X\).
    Then there exist response coefficients \(\sigma_{\infty,j}^A\) in the following sense:
    For \(\varepsilon\in\intervaloc{0,1}\) and a smooth switching function \(f\colon \R \to \intervalcc{0,1}\) satisfying \(f(t) = 0\) for \(t \le -1\), \(f(t) = 1\) for \(t \ge 0\), and \(\norm{f}_{C^{\mathfrak{C}_{n,m}}(\R)}<\Cswitch\), where
    \(
        \mathfrak{C}_{n,m}:=\ceil{m(n+1+2D/p)}
    \), consider the time-dependent Hamiltonian
    \begin{equation}
        \label{eq:thm-Hamiltonian-switching-ground-state}
        H_\epsi(t) := H_0 + \epsi \, f(t) \, V
        .
    \end{equation}
    Let \(\rho_\infty^{\epsi, \eta, f}(t)\) be the solution of the time-dependent adiabatic Schrödinger equation
    \begin{equation}
        \label{eq:thm-Schrödinger-equation-switching-ground-state}
        \I \, \eta \, \frac{\D }{\D t} \rho_\infty^{\epsi, \eta, f}(t)
        =
        \commutator[\big]{H_\epsi(t) , \rho_\infty^{\epsi, \eta, f}(t)}
    \end{equation}
    with adiabatic parameter \(\eta \in \intervaloc{0,1}\) and initial datum \(\rho_\infty^{\epsi, \eta, f}(t) = \rho_\infty\) for all \(t \le -1\).

    Then, the response to the perturbation, \(
        \Sigma_{\infty,A}^{\epsi, \eta, f}(t)
        :=
        \trace{A \, \rho_\infty^{\epsi, \eta, f}(t)}
        - \trace{A \, \rho_\infty}
    \), satisfies
    \begin{equation}
        \label{eq:thm-linear-response-ground-state}
        \sup_{\eta \in \intervalcc[\big]{\epsi^m , \epsi^{1/m}}} \abs[\bigg]{\Sigma_{\infty,A}^{\epsi, \eta, f}(t) - \sum_{j=1}^{n} \epsi^j \sigma^A_{\infty,j}}
        \le
        C \, \epsi^{n+1} \, \norm{A} \, \abs{X}^{(3+n)D} \, \paren[\big]{1 + t}^{2D/p + 1}
        ,
    \end{equation}
    for all \(t\geq 0\).

    Moreover, there exist linear maps \(\mathcal K_j\colon \alg_\Lambda \to \alg_\Lambda\) such that \(\sigma^A_{\infty,j} = \trace[\big]{\mathcal K_j(A) \, \rho_\infty}\).
    In particular \(\mathcal K_0(A)=A\),
    \(\mathcal K_1(A) = -\I \, \commutator[\big]{\calI_{H_0,g}(V),A}\), and, for \(0\le j\le n\),
    \begin{equation}
        \label{eq:K-localization}
        \norm[\big]{\mathcal K_j(A)}
        \le
        C_j \, \norm{A} \, \abs{X}^j
        ,
        \qquad
        \norm[\big]{(1-\Pi_{X_r})\mathcal K_j(A)}
        \le
        C_j \, \norm{A} \, \abs{X}^j \, \chi_{b',p'}(r)
        .
    \end{equation}
\end{theorem}

\begin{proof}
    In~\cite{HW2024}, this statement is formulated for systems which are gapped only within \(\Lambdag \subset \Lambda\).
    The first part of the statement is obtained by setting \(\Lambdap = \Lambdag = \Lambda\) from~\cite[Theorem~4.1]{HW2024}.
    Then, note that~\cite[Proposition~4.3]{HW2024} provides a sequence of \((1,p')\)-localized SLT operators \(A_\mu\) (these have nothing to do with the \(A\) used in this work).
    By the construction, \(\mathcal{K}_j\) collects the \(\epsi^j\) order terms of \(\e^{-\I \sum_{\mu=1}^n \epsi^\mu \, \commutator{A_\mu, \cdot}}\) after expanding the exponential.
    Hence, each \(\mathcal{K}_j\) is a sum of multicommutators with \(A_\mu\) of depth at most~\(j\).
    The first bound in~\eqref{eq:K-localization} therefore follows from~\cite[(93) of Lemma~A.4]{HW2024}.
    The second bound is trivial for~\(j=0\).
    For \(j\ge 1\), it follows from well-known properties of the conditional expectation~\cite[Lemma~A.6(e)]{HW2024} together with a generalization of~\cite[(92) of Lemma~A.4]{HW2024} to multicommutators, which we provide in Lemma~\ref{lem:multi-commutator-with-two-observables}.
    This concludes the proof.
\end{proof}

Now, let \(\tau_{t,s}^{\epsi,\eta,f}\) denote the Heisenberg time evolution generated by \(H_\epsi(\cdot)/\eta\), such that
\begin{equation*}
    \trace[\big]{A \, \rho_\beta^{\epsi, \eta, f}(t)}
    =
    \trace[\big]{\tau_{t,-1}^{\epsi,\eta,f}(A) \, \rho_\beta}
    \quadtext{for all}
    A\in \alg_\Lambda
    .
\end{equation*}

To proceed with the proof of Theorem~\ref{thm:response-theory-all-orders}, we combine the ground state expansion with Assumption~\nameref{ass:conv}.
Therefore, define the response coefficients by
\begin{equation*}
    \sigma_{\beta,j}^A
    :=
    \trace[\big]{\mathcal K_j(A) \, \rho_\beta}
\end{equation*}
for all \(j\in \List{1,\dotsc,n}\) with \(\mathcal K_j\) provided in Theorem~\ref{thm:linear-response-ground-states}.
With \(\mathcal K_1\) provided in Theorem~\ref{thm:linear-response-ground-states}, this in particular gives~\eqref{eq:Kubo}.
Then,
\begin{align}
    \notag
    \Sigma_{\beta, A}^{\epsi, \eta, f}(t) - \sum_{j=1}^{n} \epsi^j \, \sigma_{\beta,j}^A
    &=
    \trace[\bigg]{
        \paren[\Big]{
            \tau_{t,-1}^{\epsi,\eta,f}(A)
            - A
            - \sum_{j=1}^n \epsi^j \, \mathcal{K}_j(A)
        } \, \rho_\beta
    }
    \\&=
    \label{eq:proof-splitting-linear-response-thermal}
    \begin{aligned}[t]
        &\Sigma_{\infty, A}^{\epsi, \eta, f}(t) - \sum_{j=1}^{n} \epsi^j \, \sigma_{\infty,j}^A
        \\&+ \trace[\bigg]{
            \paren[\Big]{
                \tau_{t,-1}^{\epsi,\eta,f}(A)
                - A
                - \sum_{j=1}^n \epsi^j \, \mathcal{K}_j(A)
            } \, (\rho_\beta-\rho_\infty)
        }
    \end{aligned}
\end{align}
for all \(\beta\in \intervaloc{\beta_0,\infty}\).
The first term is bounded by the estimate~\eqref{eq:thm-linear-response-ground-state} for the response at zero temperature.
To bound the second term, denote the operator in the parenthesis~\(F\) and notice that \(\norm{F} \leq C_n \, \norm{A} \, \abs{X}^n\) by~\eqref{eq:K-localization} for some \(C_n\).
We then use the conditional expectation to decompose
\begin{equation*}
    F
    =
    \Pi_{X_r} \, F
    + (1-\Pi_{X_r}) \, F
    .
\end{equation*}
The first term \(\Pi_{X_r} \, F \in \alg_{X_r}\) is local and applying Assumption~\nameref{ass:conv}, together with \(\abs{X_r}\le\abs{X} \, (1+\Cvol \, r^D)\), gives
\begin{equation}
    \label{eq:proof-estimate-core}
    \abs[\big]{\trace[\big]{\Pi_{X_r} \, F \, (\rho_\beta-\rho_\infty)}}
    \le
    C_n \, \norm{A} \, \abs{X}^n \, \fCONV\paren[\big]{\abs{X} \, (1+\Cvol r^D)} \, \zeCONV(\beta)
    .
\end{equation}
The second term is bounded by
\begin{equation*}
    \norm{
        (1-\Pi_{X_r}) \, F
    }
    \leq
    \norm{
        (1-\Pi_{X_r})
        \, \tau_{t,-1}^{\epsi,\eta,f}(A)
    }
    + \sum_{j=1}^n
    \, \norm{
        (1-\Pi_{X_r})
        \, \mathcal{K}_j(A)
    }
    .
\end{equation*}
By Lieb-Robinson bounds and the properties of the conditional expectation, there exists \(C\),~\(\vLR>0\) such that
\begin{equation*}
    \norm[\big]{(1-\Pi_{X_r})\tau_{t,-1}^{\epsi,\eta,f}(A)}
    \le
    C \, \norm{A} \, \abs{X} \, \E^{\vLR(1+t)/\eta} \, \chi_{b',p'}(r)
    ,
\end{equation*}
for all \(r>0\), \(t\ge0\), and \(\epsi\),~\(\eta\in\intervaloc{0,1}\).
For more details see~\cite[Lemma~A.6(e) and~Lemma~A.3]{HW2024}.
Together with~\eqref{eq:K-localization}, we therefore have
\begin{equation}
    \label{eq:proof-estimate-tail}
    \norm{
        (1-\Pi_{X_r}) \, F
    }
    \leq
    C_n \, \norm{A} \, \abs{X}^n \, \E^{\vLR(1+t)/\eta} \, \chi_{b',p'}(r)
    .
\end{equation}

Adding~\eqref{eq:thm-linear-response-ground-state},~\eqref{eq:proof-estimate-core} and~\eqref{eq:proof-estimate-tail} to bound~\eqref{eq:proof-splitting-linear-response-thermal} and then minimizing over~\(r>0\) proves~\eqref{eq:linearresponse}.

\subsection*{Acknowledgments}
J.H.~acknowledges partial financial support by the ERC Advanced Grant \enquote{RMTBeyond}
No.~101020331 and the ERC Consolidator Grant \enquote{ProbQuant} (jointly with the Swiss State Secretariat for Education, Research and Innovation).
The research of M.L.~is supported by the \foreignlanguage{ngerman}{Deutsche Forschungsgemeinschaft} (DFG, German Research Foundation) through grants TRR 352 – Project-ID 470903074 and FOR 5413 - Project-ID 465199066 and by the European Union (ERC StG MathQuantProp, Grant Agreement 101163620).%
\footnote{%
    Views and opinions expressed are however those of the authors only and do not necessarily reflect those of the European Union or the European Research Council Executive Agency.
    Neither the European Union nor the granting authority can be held responsible for them.
}
This project started when T.~W.~was still at \foreignlanguage{ngerman}{Universität Tübingen} and supported by the \foreignlanguage{ngerman}{Deutsche Forschungsgemeinschaft} (DFG, German Research Foundation) –
465199066; 
470903074
.
The research of T.W.~was also partially supported by the \foreignlanguage{ngerman}{Deutsche Forschungsgemeinschaft} (DFG, German Research Foundation) under Germany’s Excellence Strategy – EXC-2111 – 390814868.

\subsection*{AI declaration}
The proof of the result and the original draft were produced in 2025 without the use of Large Language Models (LLMs).
LLMs were used to scan the draft for errors, which led to various minor improvements and corrections.
The authors take full responsibility for the manuscript in its entirety.

\appendix

\section{An auxiliary lemma}

\begin{lem}[General \(\beta\)-convergence]
    \label{lem:betaconv}
    Let \(\HSpace\) be a Hilbert space of finite dimension \(n \in \N\) and \(H \in \calL(\HSpace)\) a self-adjoint operator.
    Assume that \(H\) has non-degenerate ground state eigenvalue separated by a gap of size \(g\) from the second-smallest eigenvalue.
    Recall the notation \(\rho_\beta^{\HSpace}[H]\) for its Gibbs state at inverse temperature \(\beta \in \intervaloo{0,\infty}\).
    Then, denoting the trace norm of operators on \(\HSpace\) by \(
        \norm{\rho}_{\tr}
        :=
        \sup_{A \in \calL(\HSpace) : \norm{A} = 1}\abs{\trace{\rho \, A}}
    \) we have
    \begin{equation}
        \label{eq:convbetainf}
        \norm[\big]{\rho_\beta^{\HSpace}[H] - \rho_\infty^{\HSpace}[H]}_{\tr}
        \le
        2 \, (n-1) \, \E^{-g \beta}
        ,
    \end{equation}
    where \(\rho_\infty^{\HSpace}[H]\) denotes the projection onto the ground state of \(H\).
\end{lem}

\begin{proof}
    By spectral decomposition of \(H\) as \(H = \sum_{j=1}^n E_j \, P_j\), where \(E_j \le E_{j+1}\) are the possibly degenerate eigenvalues (we may assume w.l.o.g.~that \(E_1 = 0\)) and \(P_j\) orthogonal rank-\(1\) spectral projectors, we can write
    \begin{equation*}
        \rho_\beta^{\HSpace}[H]
        =
        \frac{\sum_{j=1}^n \E^{-\beta E_j} \, P_j}{\sum_{j=1}^n \E^{-\beta E_j}}
        .
    \end{equation*}

    We use that \(E_2 \ge g\) and find that
    \begin{align*}
        \norm[\big]{\rho_{\beta}^{\HSpace}[H] - \rho_{\infty}^{\HSpace}[H]}_{\tr}
        &=
        \sum_{k=1}^n \, \abs[\bigg]{\frac{\E^{-\beta E_k}}{\sum_{j=1}^n \E^{-\beta E_j}} - \delta_{k,1}}
        \\&=
        2 \, \frac{\sum_{j=2}^n \E^{-\beta E_j} }{\sum_{j=1}^n \E^{-\beta E_j}}
        \le
        2 \, \frac{(n-1) \, \E^{-g \beta}}{1 + (n-1) \, \E^{-g \beta}}
        \le
        2 \, (n-1) \, \E^{-g \beta}
        . \qedhere
    \end{align*}
\end{proof}

\section{Locality of multicommutators}
\label{app:multi-commutators}

\begin{lem}
    \label{lem:multi-commutator-with-two-observables}
    Let \(D\in \N\), \(\Cvol>0\), \(b > b' >0\), \(p\in \intervaloc{0,1}\) and \(k\in \N\).
    There exists a constant \(C_k>0\) such that for all lattices \(\Lambda\in \mathcal{G}(D,\Cvol)\), disjoint subsets \(X\),~\(Y \subset \Lambda\), SLT operators \(A_1,\dotsc,A_k\) and observables \(O\in \alg_X\) and \(\tilde{O} \in \Alg_Y\) it holds that
    \begin{equation}
        \label{eq:comm-SLT-two-obs}
        \norm[\big]{
            \commutator{
                \commutator{A_k, \dotsb \commutator{A_1, O} \dotsb }
                ,
                \tilde{O}
            }
        }
        \leq
        C_k \, \norm{O} \, \norm{\tilde{O}} \, \abs{X}^k \, \chi_{b',p}\pdist{X,Y}
        \, \prod_{j=1}^k \, \norm{\Phi_{A_j}}_{b,p}
        .
    \end{equation}
\end{lem}

\begin{proof}
    We write out each \(A_l\) as a sum over supports \(Z_l\) and denote \(Z_0 = X\).
    Notice that each commutator vanishes if \(Z_l \cap \bigcup_{j<l} Z_j = \emptyset\).
    Hence,
    \begin{align*}
        \Alignindent
        \norm[\big]{
            \commutator{
                \commutator{A_k, \dotsb \commutator{A_1, O} \dotsb }
                ,
                \tilde{O}
            }
        }
        \\&\leq
        \sumstack[lr]{Z_1 \subset \Lambda:\\Z_1\cap Z_0 \neq \emptyset}
        \qquad\dots\qquad
        \sumstack[lr]{Z_k \subset \Lambda:\\Z_k\cap (Z_{k-1}\cup \dotsb \cup Z_0) \neq \emptyset}
        \unit_{(Z_1\cup\dotsb\cup Z_k)\cap Y\neq\emptyset}
        \, 2^{k+1} \, \norm{\tilde{O}} \, \norm{\Phi_{A_k}(Z_k)} \dotsm \norm{\Phi_{A_1}(Z_1)} \, \norm{O}
        .
    \end{align*}
    Due to the intersecting condition, only terms satisfying \(\dist{X,Y} \leq \sum_{j=1}^k \diam{Z_j}\) contribute to this sum and we multiply each summand with
    \begin{equation*}
        \frac{\chi_{b',p}\pdist{X,Y}}{\prod_{j=1}^k \, \chi_{b',p}\pdiam{Z_j}}
        \geq
        1
        .
    \end{equation*}
    Starting from the right, we can now upper bound each sum as
    \begin{align*}
        \sumstack[lr]{Z_l \subset \Lambda:\\Z_l\cap (Z_{l-1}\cup \dotsb \cup Z_0) \neq \emptyset}
        \, \frac{\norm{\Phi_{A_l}(Z_l)}}{\chi_{b',p}\pdiam{Z_l}}
        \, \abs{Z_l}^{k-l}
        &\leq
        \sumstack{z_l\in Z_{l-1}\cup \dotsb \cup Z_0}
        \, \sumstack[r]{Z_l\subset \Lambda:\\z_l\in Z_l}
        \, \frac{\norm{\Phi_{A_l}(Z_l)}}{\chi_{b',p}\pdiam{Z_l}}
        \, \abs{Z_l}^{k-l}
        \\&\leq
        C \, \abs{Z_{l-1}\cup \dotsb \cup Z_0}
        \, \norm{\Phi_{A_l}}_{b,p}
        ,
    \end{align*}
    where the accumulating factors \(\abs{Z_l}^{k-l}\) coming from \(\abs{Z_{l-1}\cup \dotsb \cup Z_0} \leq \abs{Z_{l-1}}\, \dotsm \, \abs{Z_0}\) are absorbed in the interaction norms and only \(\abs{Z_0}^k\) remains at the end.
\end{proof}

\AtNextBibliography{\small}

\printbibliography[heading=bibintoc]

\end{document}